\documentclass[reqno,a4paper]{amsart}
\usepackage{geometry, marginnote, fancybox} 
\usepackage{graphics}  
\usepackage{tikz}
\usepackage{pgfplots}
\pgfplotsset{compat=1.3}
        
\usepackage{amsmath}
\usepackage{amsthm}
\newtheorem{thm}{Theorem}[section]
\newtheorem{lemma}[thm]{Lemma}

\newtheorem{corollary}[thm]{Corollary}
\newtheorem{remark}[thm]{Remark}

\usepackage{graphicx}
\usepackage{amssymb}
\usepackage{hyperref}
\numberwithin{equation}{section}
\numberwithin{figure}{section}
       {\begin{Sbox}\begin{minipage}}%
       {\end{minipage}\end{Sbox}\fbox{\TheSbox}}

\newcommand\Pon{\hbox{P}_{\scriptsize
       \hbox{I}}}

\newcommand\Pth{\hbox{P}_{\scriptsize
       \hbox{III}}}
\newcommand\Psix{{\rm P}_{\rm{\scriptstyle VI}}}

\newcommand\complex{\mathbb C}
\newcommand\real{\mathbb R}

\begin{document}

\title[Quicksilver solutions of $q\Pon$]{Quicksilver Solutions\\ of a {q}-difference first Painlev\'e equation }
\author{Nalini Joshi}\thanks{This research was supported by an  Australian Laureate Fellowship \#FL120100094 and grant \# DP130100967 from the Australian Research Council. NJ would also like to thank Professors Lucia di Vizio, Kenji Kajiwara and Claude M. Viallet for their helpful comments on a preliminary version of this paper.}
\address{School of Mathematics and Statistics F07, The University of Sydney, NSW 2006, Australia}
\email{nalini.joshi@sydney.edu.au}
\begin{abstract}
In this paper, we present new, unstable solutions, which we call quicksilver solutions,  of a $q$-difference Painlev\'e equation in the limit as the independent variable approaches infinity. The specific equation we consider in this paper is a discrete version of the first Painlev\'e equation ($q\Pon$), whose phase space (space of initial values) is a rational surface of type $A_7^{(1)}$. We describe four families of almost stationary behaviours, but focus on the most complicated case, which is the vanishing solution. We derive this solution's formal power series expansion, describe the growth of its coefficients and show that, while the series is divergent, there exist true analytic solutions asymptotic to such a series in a certain $q$-domain. The method, while demonstrated for $q\Pon$, is also applicable to other $q$-difference Painlev\'e equations. 
\end{abstract}

\maketitle

\section{Introduction}
We consider asymptotic behaviours of almost stationary solutions of the $q$-difference first Painlev\'e equation (derived by Ramani {\it et al.} \cite{rg:96})
\begin{equation}\label{eq: qp1 scaled}
{q}\Pon :\quad  w_{n+1}\,w_{n-1} =\frac{1}{w_n}-\,\frac{1}{\xi\, w_n^2}
\end{equation}
as $\xi\to\infty$, where $\xi=\xi_0\,q^n$, $w_n=w(\xi)$, $w_{n+1}=w(q\,\xi)$, $w_{n-1}=w(\xi/q)$ (with $q\not=0, 1$). (See \S \ref{sec: notation} for further details on notation.)
In the continuum limit, $q\Pon$ becomes the first Painlev\'e equation $\Pon$: $y_{tt}=6\,y^2-t$ (see \S \ref{sec: versions}). In this paper, we identify new unstable solutions that vanish in the limit $|\xi|\to\infty$, which we call \lq\lq quicksilver solutions\rq\rq . 

The six classical Painlev\'e equations (denoted $\Pon$--$\Psix$) were discovered by Painlev\'e and co-workers in their search for second-order ordinary differential equations with solutions that can be globally continued in the complex plane \cite{p:02,f:05,g:10}. There has been significant interest in these equations due to their role as physical models in a wide range of applications. Their solutions are regarded as modern non-linear special functions \cite{c:10} and certain critical solutions have attracted a great deal of attention as  universal models in quantum gravity \cite{gm:90} and random matrix theory \cite{d:07}. 

More recently, considerable attention has been directed at discrete versions of the Painlev\'e equations. These discrete versions arise in orthogonal polynomial theory \cite{s:39} and in studies of quantum gravity \cite{fik:91}. Many discrete Painlev\'e equations were discovered by using the singularity confinement property \cite{rgh:91}, which was regarded as a discrete manifestation of the property used by Painlev\'e and his colleagues in the context of ordinary differential equations. For each continuous Painlev\'e equation, there are now known to be many integrable second-order discrete versions.

By showing how to obtain discrete Painlev\'e equations as mappings on rational surfaces obtained from a nine-point blow-up of the complex projective plane, Sakai \cite{s:01} resolved and classified these equations. $q\Pon$ is identified in Sakai's classification as an equation associated with a rational surface of type $A_7^{(1)}$. The iteration of the equation occurs as a mapping of points on the rational surface and, therefore, the latter is referred to as the initial value space of the equation. In \S \ref{sec: versions}, we show how the other $q$-discrete versions of $\Pon$ given in the literature are related to Equation \eqref{eq: qp1 scaled}. 

Although the equations are now classified and relations between them understood, no analytic information about any generic solution of a $q$-difference version of any Painlev\'e equation has yet been found. Nishioka \cite{n:10} has proved that $q\Pon$ cannot have any solutions that are expressible in terms of solutions of first-order difference equations and therefore that its solutions cannot be expressed in terms of well known basic or $q$-special functions that satisfy linear difference equations. The current paper investigates the analytic properties of the higher transcendental solutions of $q\Pon$ for the first time.

Our investigation is analogous to the asymptotic study of the solutions of the first Painlev\'e equation in the limit as $|t|\to\infty$, which began a century ago \cite{b:13}. Whilst the general solution of $\Pon$ is meromorphic, with poles distributed across the complex plane, it is well known that there exist solutions that appear to have no poles near infinity in a sector of angular width $4\pi/5$. These are one-free-parameter families of solutions asymptotic to $(t/6)^{1/2}\bigl(1+{\mathcal O}(t^{-3/4(1+\epsilon)})\bigr)$ in such sectors, called \textit{tronqu\'ee} solutions\footnote{The name arises because the lines of poles of such solutions are truncated in these sectors.}.  Within these families, there exist five unique solutions called tri-tronqu\'ee solutions, which are asymptotically free of poles in a sector of width $8\pi/5$. Only one of these is real on the real line and this solution was shown to have no poles whatsoever on the real positive semi-axis in \cite{jk:01}. Another distinguishing feature of such solutions is that they arise in neighbourhoods of singularities, i.e., points where the gradient vanishes, of the leading-order conserved quantity \cite{dj:11}.

The solutions we describe in this paper share comparable properties, but differ in one significant respect. For $q\in{\complex}$, solutions are iterated on $q$-spirals in the $\xi$-plane. (We assume that $q$ is complex, but note that when $q$ is real, these spirals are replaced by semi-infinite real intervals.) Whilst its region of asymptotic validity is not sectorial in the traditional sense, it is a region with boundaries given by arcs of spirals, such as shown in Figure \ref{spiral sectors}. Because the solution's asymptotic behaviour is given by a power series expansion in such a region, it is  bounded and free of poles there. So far, these properties are analogous to those of the truncated solutions in the continuous case. 
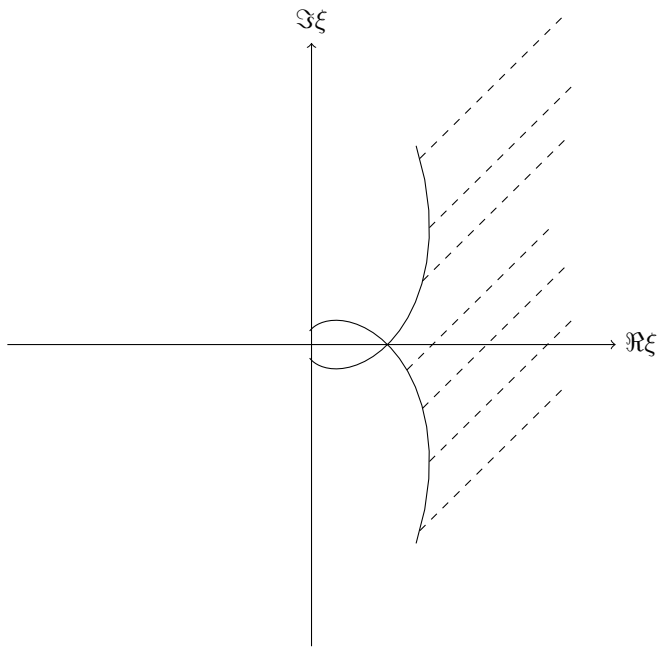
\begin{figure}[t]
\begin{tikzpicture}[auto]
  \draw[->] (-4,0) -- (4,0) node[right] {$\Re{\xi}$};
  \draw[->] (0,-4) -- (0,4) node[above] {$\Im{\xi}$};
  \draw   plot [domain=-1.707:1.09] ({ exp(\x)*cos(\x r) }, { exp(\x)*sin(\x r) })  ;
  \draw   plot [domain=-1.09:1.707] ({ exp(-\x)*cos(\x r) }, { exp(-\x)*sin(\x r) })  ;
  \draw[dashed]   plot [domain=0:1.9] ({ exp(0.785398)*cos(0.785398 r) +\x }, { exp(0.785398)*sin(0.785398 r) + \x })  ;
  \draw[dashed]   plot [domain=0:1.9] ({ exp(0.523599)*cos(0.523599 r)+ \x }, { exp(0.523599)*sin(0.523599 r)+ \x })  ;
  \draw[dashed]   plot [domain=0:1.9] ({ exp(1.047198)*cos(1.047198 r)+ \x }, { exp(1.047198)*sin(1.047198 r)+ \x })  ;
  \draw[dashed]   plot [domain=0:1.9] ({ exp(0.785398)*cos(0.785398 r) +\x }, { - exp(0.785398)*sin(0.785398 r) + \x })  ;
  \draw[dashed]   plot [domain=0:1.9] ({ exp(0.523599)*cos(0.523599 r)+ \x }, { -exp(0.523599)*sin(0.523599 r)+ \x })  ;
  \draw[dashed]   plot [domain=0:1.9] ({ exp(1.047198)*cos(1.047198 r)+ \x }, { -exp(1.047198)*sin(1.047198 r)+ \x })  ;
  \draw[dashed]   plot [domain=0:1.9] ({ exp(0.2617994)*cos(0.2617994 r)+ \x }, { -exp(0.2617994)*sin(0.2617994 r)+ \x })  ;
  \end{tikzpicture}\caption{An example of a sectorial region bounded by arcs of spirals} \label{spiral sectors}
\end{figure}

However, a major difference with the continuous case arises when we consider how these asymptotic behaviours are related to the singularities of the invariant of the leading order autonomous equation. 
Consider $q\Pon$ to leading-order:
\begin{equation}\label{eq: qp1 auto}
 \mathtt{w}_{n+1}\,\mathtt{w}_n\,\mathtt{w}_{n-1} =1
\end{equation}
This equation has an invariant given by
\begin{equation}\label{eq: ham auto}
K(x, y) =\frac{x^2\,y^2+x\,y+x+y}{x\,y}
\end{equation}
i.e., $K(\mathtt{w}_{n+1}, \mathtt{w}_{n})-K(\mathtt{w}_{n}, \mathtt{w}_{n-1})=0$ when $\mathtt{w}_n$ satisfies Equation \eqref{eq: qp1 auto}. The only finite points where the gradient of $K(x, y)$ vanishes is given by $x=y=\omega$ where $\omega^3=1$. 
Three of the solutions we find lie in the neighbourhood of these singularities of $K(x, y)$, but the fourth one is very different. 

That remaining, vanishing, solution is asymptotically close to a base point $(w_n, w_{n-1})=(1/\xi, 0)$, where the flow of the $q\Pon$ vector field becomes undefined (i.e., the definition of $w_{n+1}$ becomes $0/0$). It, therefore, has the most complicated behaviour and we focus on the properties of this solution in this paper. This solution is unstable under perturbations in the initial-value space.  Due to its mercurial, unstable nature and distinguished role in the space of initial values, we named it a \textit{quicksilver} solution. We deduce an asymptotic series expansion for this behaviour and show that the series is divergent. Nevertheless, we are able to prove that a continuous solution exists, which is asymptotic to the series. We call such a solution a \lq\lq true solution\rq\rq\ in this paper.

Such proofs are well developed for ordinary differential equations \cite{hs:99} and for classes of difference equations  \cite{hs:64} for which the equation can be expanded in a power series in the independent variable $n$ in the limit of interest. This is not the case for $q\Pon$, even after transforming to a new variable $n=\log\xi/\log q$ in which the iterates appear in standard form.  This is due to the fact that the term $\xi=\xi_0\,q^n$, which appears explicitly in $q\Pon$, cannot be expanded in a power series in $n$ as $n\to\infty$ with $|q|>1$. 

Other common techniques in the literature rely on Borel transforms to sum asymptotic series. However, to our knowledge, the application of this technique to non-autonomous $q$-difference equations, and therefore corresponding results for existence of true solutions that are asymptotic to a divergent series, remains restricted to {\em linear} $q$-difference equations (see references in Ramis {\em et. al}\/\cite{rsz:12}). 

We overcome these restrictions by using an idea of Cotton \cite{c:11}, adapted to the setting of $q$-difference equations. In the following, we use the phrase \lq\lq almost invariant solution\rq\rq\ to mean solutions that approach a steady-state as $|\xi|\to\infty$. The plan of the paper is as follows. In \S \ref{expansion}, we derive formal series expansions of almost invariant solutions, which satisfy $q\Pon$ in the asymptotic limit $|\xi|\to\infty$. We focus on the vanishing solution and describe the growth of coefficients of its formal series expansion as the summation index approaches infinity. In \S \ref{true}, we prove the existence of true solutions asymptotic to such series. The paper ends with a conclusion in \S \ref{conclusion}. Details of the proofs in \S\S \ref{expansion} and \ref{true} are provided in the appendices. 

\subsection{Notation}\label{sec: notation} Several notations are used throughout the paper to denote the same function: $w=w(\xi)=w(\xi_0\,q^n)=w_{n}$, $\overline w=w(q\,\xi)=w(\xi_0\,q^{n+1})=w_{n+1}$, $\underline w=w(\xi/q)=w(\xi_0\,q^{n-1})=w_{n-1}$. Upper and lower bars may also be used to denote iteration of functions of other variables, where there is no ambiguity.

\subsection{Different Versions of $q\Pon$}\label{sec: versions}
There are different versions of the $q$-discrete first Painlev\'e equation with rational surface of type $A_7^{(1)}$ in the literature. Historically, the first one was derived as a limiting case of the $q$-discrete third Painlev\'e equation or q$\Pth$ (associated with rational surface $A_5^{(1)}$) by Ramani {\it et al}\/\cite{rg:96}. In this derivation, the equation is given as
\begin{equation}\label{eq: qp1 z}
y_{n+1} \, y_{n-1} = \frac{z}{y_n}+\frac{\mu}{y_n^2}
\end{equation}
where $z=\zeta_0\lambda^{n/2}$ and $\mu$ is constant. Under the identification $\nu=\lambda^{1/2}$, $x=\nu^n$, $g(x)=y_n$, and $\zeta_0=\alpha$, $\mu=\beta$, this equation becomes  
\begin{equation}\label{eq: qp1 x}
g_{n+1} \, g_{n-1} = \frac{\alpha\,x}{g_n}+\frac{\beta}{g_n^2}
\end{equation}
where $g_n=g(x)$, $g_{n+1}=g(\nu\,x)$, $g_{n-1} = g(x/\nu)$. The continuum limit of Equation \eqref{eq: qp1 z} is given by $y_n=1+\epsilon^2\,y(t)$, $z=4\,\lambda^{n/2}$, $\mu=-3$, $\lambda=1-\epsilon^5/8$, $t=-\epsilon\,n$, which leads to $y_{tt}=6\,y^2-t$. 

Another version is
\begin{equation}\label{eq: qp1 a7}
 f_{n+1} \, f_{n-1}= -\,\frac{s}{f_n}+\,\frac{s}{f_n^2}
\end{equation}
where $s$ is an exponential function of $n$ and $f_n=f(s)$. This equation is due to Sakai \cite{s:01} and studied in  \cite{n:10,o:10}. It can be transformed to Equation \eqref{eq: qp1 x}, by taking a new dependent variable: $f_n=v_n/\rho_n$, which leads to 
\begin{equation}\label{eq: qp1 w}
v_{n+1}\,v_{n-1} = \rho_{n+1}\,\rho_n\,\rho_{n-1}\,\left(-\,\frac{s}{v_n}+\frac{s\rho_n}{v_n^2}\right)
\end{equation}
This becomes Equation \eqref{eq: qp1 x}, when we take $\rho_n=-\,\beta/(\alpha\,x)$ and $s=\alpha^4\,x^4/\beta^3$, with $v_n=v(s)=g(x)=g_n$.

We take a further transformation, useful for studying the asymptotic limit $x\to\infty$, by using 
\begin{align}
\label{eq: x transf} g_n&=x^{1/3}\,u_n\ 
\Rightarrow\ u_{n+1}\,u_{n-1} =\frac{\alpha}{u_n}+\frac{\beta}{x^{4/3}\,u_n^2}
\end{align}
Defining $w_n=w(\xi)=u_n/\alpha^{1/3}$, $\xi=-\,\gamma\,x^{4/3}$, where $\gamma=\alpha^{4/3}/\beta$, and renaming $\nu^{4/3}=q$, we arrive at Equation \eqref{eq: qp1 scaled}, which is the subject of this paper.

\section{Asymptotic Expansions of Almost Invariant Solutions}\label{expansion}
In this section, we consider solutions of $q\Pon$ whose leading order behaviour as $|\xi|\to\infty$ is independent of $n$. Detailed proofs of major results in this section can be found in Appendix \ref{app expansion}. 
The assumption of stationarity (to leading order) is equivalent to $\overline w\sim w$ and $\underline w\sim w$ as $|\xi|\to\infty$. To leading order, stationarity in Equation \eqref{eq: qp1 scaled} implies that $w(w^3-1)={\mathcal O}(1/\xi)$. Therefore, we obtain either $w^3=1+{\mathcal O}(1/\xi)$ or $w={\mathcal O}(1/\xi)$. These two cases are considered in separate subsections below. As remarked in the Introduction, we focus more on the details of the second case in this paper.

\subsection{Non-zero Asymptotic Behaviour}
For the case $w^3\sim 1$, we obtain the following result.
\begin{lemma}
Equation \eqref{eq: qp1 scaled} has a formal power series solution
\begin{equation}\label{eq: formal}
w(\xi)=\sum_{n=0}^\infty\frac{a_n}{\xi^n},
\end{equation}
where 
\begin{subequations}\label{formal coefficients}
\begin{align}
&a_0^3=1, \label{formal rec 0}\\
&a_1\,( q+1+ q^{-1})=-\,1,\label{formal rec 1}\\
&\sum_{m=0}^{n}\sum_{j=0}^{n-m}\sum_{l=0}^m\,a_j\,a_{n-m-j}\,a_l\,a_{m-l} q^{(n-m-2\,j)}=a_n,\ n\ge 2. \label{formal rec n}
\end{align}
\end{subequations}
\end{lemma}
\begin{proof}
Multiply Equation \eqref{eq: qp1 scaled} by $w^2$ and use Lemma \ref{convolution lemma}.
\end{proof}
\begin{remark}
Let $\omega$ be a cubic root of unity, $\omega^3=1$. Equation \eqref{eq: qp1 scaled} is invariant under the discrete rotation $w\mapsto \omega w$, with $\xi\mapsto \xi/\omega$. Hence, we restrict our attention to the case $a_0=1$ without loss of generality.
\end{remark}
\begin{remark}
The recursion relation \eqref{formal rec n} can be rearranged to yield a formula for the $n$-th coefficient $a_n$. Doing so yields
\begin{align}
\nonumber a_n&\bigl( q^n+1+ q^{-n}\bigr)\\
\nonumber &=-\,\sum_{l=1}^{n-1}a_l\,a_{n-l}\,( q^{(2l-n)}+1)\\
\label{a_n}&\quad -\,\sum_{m=1}^{n-1}\sum_{j=0}^{n-m}\,\sum_{l=0}^m\,a_j\,a_{n-m-j}\,a_l\,a_{m-l}\, q^{(n-m-2\,j)}
\end{align}
\end{remark}
Note that the recursion relation \eqref{a_n} is symmetric under $q\mapsto 1/q$. 
\subsection{Vanishing Asymptotic Behaviour}
For the case $w\ll 1$ as $\xi\to\infty$, we obtain the following result.
\begin{lemma}
Equation \eqref{eq: qp1 scaled} has a formal power series solution
\begin{equation}\label{eq: formal z}
w(\xi)=\sum_{n=1}^\infty\frac{b_n}{\xi^n},
\end{equation}
where 
\begin{subequations}\label{z formal coefficients}
\begin{align}
&b_1=1 \label{z rec 1}\\
&b_2=0 \label{z rec 2}\\
&b_3=0 \label{z rec 3}\\
&b_n=\sum_{r=2}^{n-2}\sum_{k=1}^{r-1}\sum_{m=1}^{n-r-1}\,b_k\,b_{r-k}\,b_m\,b_{n-r-m} q^{(r-2\,k)},\ n\ge 4 \label{z formal rec n }
\end{align}
\end{subequations}
\end{lemma}
\begin{proof}
Multiply Equation \eqref{eq: qp1 scaled} by $w^2$ and use Lemma \ref{convolution lemma}. Note that the first term in the product $\overline w\,w^2\,\underline w$ is given by $b_1^4/\xi^4$. Hence Equations (\ref{z rec 1}-\ref{z rec 3}) follow. Equation \eqref{z formal rec n } follows from the remaining terms in the convolution.
\end{proof}
The following two results about the coefficients $b_n$ are proved in Appendix \ref{app expansion}.

\begin{corollary}\label{modulo 3}
$b_{3\,p+2}=0$ and $b_{3\,p+3}=0$ for all integers $p\ge 0$.
\end{corollary}
The recursion relation for the non-zero coefficients $d_p=b_{3\,p+1}$ is given by Equation \eqref{staggered rec} in Appendix \ref{app expansion}. The first few values of $d_p$ are
\begin{subequations}
\begin{align}
d_1&=1\\
d_2&=q^3+2+\frac{1}{q^3}\\
d_3&=q^9+2\,q^6+5\,q^3+6+\frac{5}{q^3}+\frac{2}{q^6}+\frac{1}{q^9}\\
d_4&=q^{18}+2\,q^{15}+5\,q^{12}+10\,q^9+16\,q^6+23\,q^3+26+\frac{23}{q^3}+\frac{16}{q^6}+\frac{10}{q^9}+\frac{5}{q^{12}}+\frac{2}{q^{15}}+\frac{1}{q^{18}}
\end{align}
\end{subequations}

\begin{lemma}\label{z large n}
For $|q|>e^2$, the solution of the recurrence relation \eqref{z formal rec n } has the asymptotic behaviour \begin{equation}
b_{3\,p+1}\underset{p\to\infty}{=}{\mathcal O}\left(|q|^{3\,p\,(p-1)/2}\,\prod_{k=0}^{p-1}\bigl(1+q^{-3k}\bigr)^2\right).
\end{equation}
\end{lemma}
\noindent This lemma shows that the coefficients of the asymptotic expansion \eqref{eq: formal z} grow as a quadratic function of its index. The series is, therefore, divergent. We provide a direct proof of this lemma in Appendix \ref{app expansion}. It should be noted that the bound on $|q|$ can be improved to $|q|>1$. The assumption  $|q|>e^2$ is made here to obtain simpler explicit estimates in the proof. 
\begin{remark}
Note that an estimate of the growth of the coefficients could also have been obtained from an application of the main theorem in Zhang \cite{z:1998}. The latter requires properties of the linearisation of $q\Pon$ (considered in the next section) and shows that the rescaled series $\sum b_n |q|^{-\,3\,n\,(n+1)/2}/\xi^n$ converges. In the terminology of that paper, the formal series \eqref{eq: formal z} lies in the space of q-Gevrey series of order $3$. However, the result of Lemma \ref{z large n} provides a sharper estimate of growth of the coefficients $b_n$ because when $n=3p+1$, the q-Gevrey estimate becomes $|q|^{27\,(p+1/3)\,(p+2/3)/2}$, which is a weaker upper bound.
\end{remark}
\begin{remark}
In the case $|q|<1$, our arguments would have led to $b_{3\,p+1}={\mathcal O}\bigl(|q|^{-3\,p\,(p-1)/2}\prod_{k=0}^{p-1}\bigl(1+q^{3k}\bigr)^2\bigr)$.
\end{remark}
\section{True Solutions}\label{true}
In this section, we prove that there exist continuous solutions of Equation \eqref{eq: qp1 scaled} in a non-empty region of the $\xi$-plane, which are asymptotic to the formal series expansion \eqref{eq: formal z} deduced in the previous section. Recall that we refer to such solutions as \lq\lq true solutions\rq\rq . Also, note that we  assume $|q|>1$.

As noted in \S 1, for $q\in\complex$, the typical region of asymptotic validity for $q$-difference equations is no longer a sector, but regions bounded by $q$-spirals in the complex plane. Let $D_0$ be a disk in $\complex$ whose points are at a distance at least $r_0>0$ away from the origin. Define $D$ to be a union of open sets $D_n=q^n\,D_0$. There is a finite collection of sectors ${\mathcal S}_k=\bigl\{\xi\in\complex \ \bigm|\ r_0<1/|\xi|\le r_k, \alpha_k\le \arg{\xi}\le \beta_k\bigr\}$ in $\complex^*$ such that each $D_n$ is contained in one of the ${\mathcal S}_k$. Let ${\mathcal S}=\cup {\mathcal S}_k$.

There are two main steps in the proof. First, given $\mathcal S$, a standard use of the Borel-Ritt theorem \cite{f:53,hs:99} provides a function $W(\xi)$ such that
\begin{enumerate}
\item $W(\xi)$ is analytic in $\mathcal S$;
\item $W(\xi)\sim \sum_{n=1}^\infty\,{b_n}/{\xi^n}$, as $\xi\to\infty$ in $\mathcal S$.
\end{enumerate} 
Second, we show that there exists a true solution of $q\Pon$ asymptotic to \eqref{eq: formal z} by using a contraction mapping argument. Proofs of the following results are provided in Appendix \ref{app true}.

\begin{lemma}\label{perturbation} Under the change of variables to $w_n= W_n+\,v_n$, Equation \eqref{eq: qp1 scaled} becomes
\begin{equation}\label{eq: qp2 pert}
v_{n+1}+\left(2\frac{W_{n+1}}{W_n}-\,\frac{1}{W_n^2\,W_{n-1}}\right)\,v_n+\frac{W_{n+1}}{W_{n-1}}\,v_{n-1}={\mathcal R}_2(v_n, v_{n-1}, t)
\end{equation}
where $t=1/\xi$ and ${\mathcal R}_2$ is holomorphic in the domain 
\[{\mathcal D}=\bigl\{(x, y, t)\in \complex^3\Bigm| |t|<\epsilon_0, \|(x,y)\|\le \delta_0\bigr\}.\]
for some positive real numbers $\epsilon_0$ and $\delta_0$. Moreover, 
\begin{align*}
{\mathcal R}_2(0, 0, t) ={\mathcal O}\bigl(|t|^N\bigr),\ \forall N\in{\mathbb N},\quad
\frac{\partial {\mathcal R}_2}{\partial x} (0, 0, t)=0,\ 
\frac{\partial {\mathcal R}_2}{\partial y} (0, 0, t)=0 .
\end{align*}
Letting ${\mathbf v}=(x, y)$, define the norm $\|{\mathbf v}\|=\max\{|x|, |y|\}$ and let $\nabla {\mathcal R}_2=(\partial_{x}{\mathcal R}_2, \partial_{y}{\mathcal R}_2)$. Then there exist strictly positive real numbers $\delta_0$, $\epsilon_0$, $C_1$, $C_2$, $C_3$, and $C_4$ such that 
\begin{align*}
&\|{\mathcal R}_2(x, y, t)\|\le C_1\,\|{\mathbf v}\|^2+C_2|t|\\
&\|\nabla {\mathcal R}_2\| \le C_3\,\|{\mathbf v}\|+C_4\,|t|
\end{align*}
if $\|{\mathbf v}\|<\delta_0$ and $|t|<\epsilon_0$. 
\end{lemma} 
In the next lemma, we find solutions of the linear homogeneous part of Equation \eqref{eq: qp2 pert}.  
\begin{lemma}\label{linear eq}
Let $P_n$ be a solution of the linear homogeneous difference equation
\begin{equation}\label{eq: lin hom P}
P_{n+1} + \,\left(2\,\frac{W_{n+1}}{W_n}-\frac{1}{W_n^2\,W_{n-1}}\right)\,P_n+\frac{W_{n+1}}{W_{n-1}}\,P_{n-1}=0
\end{equation}
Then, there exist two solutions $P^\pm(\xi)$ of this equation with the following asymptotic behaviours valid as $n\to\infty$
\begin{subequations}\label{pasym}
\begin{align}
\label{p plus}P^+(\xi)&\sim \prod_{j=n_0}^n\frac{1}{W_j^2\,W_{j-1}}\sim c_+\,q^{3\,n(n-5/3)/2} \\
\label{p minus}P^-(\xi)&\sim \prod_{j=n_0}^n\,W_{j+1}\,W_j^2\sim c_-\,q^{-\,3\,n(n+5/3)/2} 
\end{align}
\end{subequations}
for some constant initial point $n_0$ and constants $c_\pm$. 
\end{lemma}
\begin{remark}
Since the functions $P^{\pm}$ satisfy a linear homogeneous equation \eqref{eq: lin hom P}, we assume $c_{\pm}=1$ without loss of generality.
\end{remark}
\begin{remark}
The Newton polygon of \eqref{eq: lin hom P} is an equilateral triangle with its base on the interval $[0, 2]$ and remaining sides of slope $3$ and $-3$ respectively. See the definition of Newton polygons and associated theory of asymptotic behaviour of solutions in \S 2.2 of \cite{rsz:12}.
\end{remark}
\begin{remark}\label{P size}
The two solutions $P^{\pm}$ are distinguished by their size in the limit as $n\to\infty$ only in certain regions of the $\xi$-plane. The boundaries of these Stokes-like regions are given by $|P^{\pm}/P^{\mp}|={\mathcal O}(1)$. Letting $\xi=\xi_0\, q^{n}$, and writing $\xi/\xi_0=r\,e^{i\,\theta}$, we find $n=\ln (\xi/\xi_0)/\ln q$ and so the boundaries are given by $(\ln r)^2=\theta^2$. An example of such (anti-) Stokes-like boundaries  (assuming $-\,\pi< \theta\le \pi$) is shown in Figure \ref{fig:stokes}. We denote the region where $|P^{+}/P^{-}|\gg 1$ by $\mathcal E$ and work henceforth in the intersection $\mathcal S\cap \mathcal E$, which we denote once more by the symbol $\mathcal S$. 
\end{remark}
\begin{figure}[t]
\begin{tikzpicture}[auto]
  \draw[->] (-5,0) -- (2,0) node[right] {$\Re{\xi}$};
  \draw[->] (0,-2) -- (0,2) node[above] {$\Im{\xi}$};
  \draw   plot [domain=-3.14159:3.14159,smooth] ({ exp(\x)*cos(\x r)/5 }, { exp(\x)*sin(\x r)/5})  ;
 \draw   plot [domain=-3.14159:3.14159,smooth] ({exp(-\x)*cos(\x r)/5}, { exp(-\x)*sin(\x r)/5})  ;
 \coordinate [label=above:{$r=e^{\theta}$}] (plus) at (-1.5,1.5);
 \coordinate [label=below:{$r=e^{-\theta}$}] (minus) at (-1.5,-1.5);
  \end{tikzpicture}\caption{A Stokes boundary where $|P^{+}/P^{-}|={\mathcal O}(1)$, shown in the $\xi$-plane with $\xi=r\,e^{i\,\theta}$, where $-\pi< \theta\le \pi$. (See Remark \ref{P size}.) The exterior $\mathcal E$ of the combined spirals shown here represents the region where $|P^+|\gg |P^{-}|$.} \label{fig:stokes}
\end{figure}
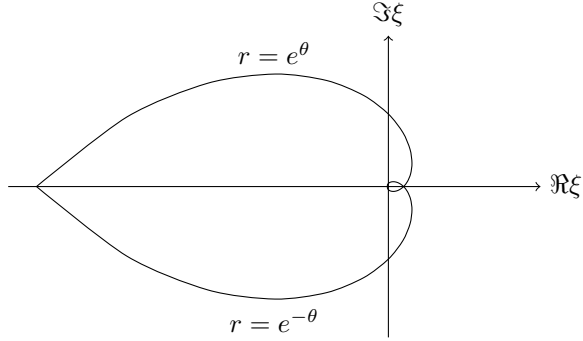
\begin{lemma}\label{summation eq}
Any solution $v$ of Equation \eqref{eq: qp2 pert} satisfies the summation equation
\begin{align}\label{eq: v int}
v_n=\alpha_0\,P_n+P_n\,\sum_{j=n_0+1}^n\,\frac{W_j\,W_{j-1}}{P_j\,P_{j-1}}\,\sum_{k=n_0+1}^{j}\frac{P_k\,{\mathcal R}_2(v_k, v_{k-1},t_k)}{W_{k+1}\,W_{k}}.
\end{align}
where $P_n$ is a solution of the linear homogeneous equation \eqref{eq: lin hom P} and $\alpha_0$, $n_0$ are constants. If we use the convention that the summation is zero when the upper bound is less than the lower bound, then $\alpha_0=v_{n_0}/P_{n_0}$.
\end{lemma}
\begin{corollary}
The proof of the above result also gives the equivalent summation equation:
\begin{align}\label{eq: v int reverse}
v_n=\beta_0\,P_n-\,P_n\,\sum_{j=n}^{n_0-1}\,\frac{W_j\,W_{j-1}}{P_j\,P_{j-1}}\,\sum_{k=k_0}^{j-1}\frac{P_k\,{\mathcal R}_2(v_k, v_{k-1},t_k)}{W_{k+1}\,W_{k}}
\end{align}
where $k_0$ and $n_0$ are constants, with $\beta_0=v_{n_0}/P_{n_0}$.
\end{corollary}
An application of the contraction mapping theorem to the summation equation \eqref{eq: v int reverse} gives a desired true solution, along a semi-infinite spiral path ${\mathcal T}$, as shown by the following result.  
\begin{lemma}\label{contraction}
Suppose constants $\delta_0$, $\epsilon_0$, $C_1$, $C_2$, $C_3$, and $C_4$ are as found in Lemma \ref{perturbation}. Let $0<\delta<\delta_0$, $0<\epsilon<\epsilon_0$, $t_0$ be complex numbers, with $|t|<\epsilon$, for every $t\in {\mathcal T}$, where $t_0\in D_0$ a disk in $\complex$, $n_0=-\ln(|t_0|)/\ln(|q|)>2$, and ${\mathcal T}=t_0\,q^{-\,(n_0+\real^+)}$, satisfying $|q|\bigl(C_1\delta^2+C_2\epsilon\bigr) < \delta$ and $|q|\bigl(C_3\,\delta+ C_4\,\epsilon\bigr) <1$, with $1/t\in\mathcal S$ as defined in Remark \ref{P size}. Then there exists a solution $v(t)$ of Equation \eqref{eq: qp2 pert}, which is defined for all $t\in {\mathcal T}$ and satisfies $\|v(t)\|\le \delta$ for all $t\in {\mathcal T}$.
\end{lemma} 
\begin{remark}
Similar conclusions hold when the path ${\mathcal T}$ is replaced by a curve $t_0\,q^{-\,\tau+i\,\sigma(\tau)}$, where $\tau\in [n_0, \infty)$, $\sigma: [n_0, \infty)\mapsto \real$ is continuously differentiable with bounded derivative and the curve stays sufficiently far away from the origin with $1/t\in\mathcal S$.
\end{remark}
\begin{remark}
The linear equation \eqref{eq: lin hom P} is the linear part of the perturbation equation of a solution $w$ of $q\Pon$ that has leading-order behaviour $1/\xi$. That is, the generic perturbations of the solution found in Lemma \ref{contraction} will involve linear combinations of $P^\pm$. Since one of these grows exponentially fast in $\xi$, it follows that the solution of Lemma \ref{contraction} is unstable.
\end{remark}

\section{Conclusion} \label{conclusion}
In this paper, we deduced asymptotic expansions of four families of solutions $w(\xi)$ of $q\Pon$ \eqref{eq: qp1 scaled} that are stationary to leading-order in the limit $\xi\to\infty$. In three of these cases, the solution satisfies $w_n\sim \omega$, where $\omega^3=1$, whilst in the remaining case it satisfies $w_n\sim 1/\xi$, as $\xi\to\infty$. The cases when the solutions approach a cube root of unity correspond to singularities of the invariant $K$ (see Equation \eqref{eq: ham auto}), while the vanishing case plays an altogether different role in the space of initial values.

The forward iterate $w_{n+1}$, given by
\[w_{n+1}=\frac{w_n-1/\xi}{w_{n-1} \,w_n^2}\]
remains well defined as $w_n\to \omega$ and $w_{n-1}\to\omega$.  However, when $w_n\to 1/\xi$ and $w_{n-1}\sim q/\xi$, $w_{n+1}$ approaches $0/0$. That is, the solution approaches a base point, where the dynamics of the flow defined by $q\Pon$ become ambiguous. The flow around an isolated such point could have been resolved by standard techniques from algebraic geometry, however, in this case, two base points (one from the forward iteration and one from backward iteration) coalesce in the asymptotic limit, creating a more complicated flow near the origin.

Because of this feature, we focused on the vanishing behaviour in this paper. We showed that the corresponding asymptotic series is divergent by deducing the behaviour of its coefficients as the summation index becomes large. Furthermore, we proved that for initial values on each sufficiently large smooth path in $\mathcal S$, there exists a true vanishing solution, unstable in initial-value space, which is asymptotic to the formal power series. Its instability, and existence near a coalescing pair of base points, led us to allocate the solution a distinguished name, i.e. {\it quicksilver} solution, in the sense of being quickly changeable and difficult to hold or contain.
\appendix
\section{Proofs of Results in \S\ref{expansion}}\label{app expansion}
In the derivation of the recursion relations in \S \ref{expansion}, we use the following standard result on products of formal series.
\begin{lemma}\label{convolution lemma}
Suppose $F=\sum_{k=k_0}^\infty F_k/\xi^k$ and $G=\sum_{m=m_0}^\infty G_m/\xi^m$, where $k_0$ and $m_0$ are non-negative integers. Then the product $F\,G$ has the series expansion
\[
F\,G=\sum_{n=k_0+m_0}^\infty\,\frac{H_n}{\xi^n}
\] 
where $H_k=\sum_{k=k_0}^{n-m_0}\,F_k\,G_{n-k}$.
\end{lemma}

\subsection{Proof of Corollary \ref{modulo 3}}
\begin{proof}
We prove the result by induction. The case $p=0$ is already true. Now consider $p\ge 1$. Assume $b_{3\,i+2}=0$ and $b_{3\,i+3}=0$ for all integers $0\le i\le p-1$ and consider the inductive step $i=p$. Denote each of the convolution sums in Equation \eqref{z formal rec n } respectively by 
\begin{subequations}
\begin{align}
\label{convolution A}
A_{r}&=\sum_{k=1}^{r-1}\,b_k\,b_{r-k} q^{(r-2\,k)}\\
\label{convolution B}
B_{n-r}&=\sum_{m=1}^{n-r-1}\,b_m\,b_{n-r-m}
\end{align}
\end{subequations}
Note that we have $b_{n}=\sum_{r=2}^{n-2}A_r\,B_{n-r}$.

Consider the case $n=3\,p+2$. For a  term $b_k\,b_{r-k}$ of $A_r$ to be non-zero, there must exist integers $j$ and $l$,  such that $k=3\,j+1$ (satisfying $1\le 3\,j+1\le r-1$) and $r-k=3\,l+1$ (satisfying $2\le 3\,(j+l)+2\le n-2$). Adding these two equations, we obtain $r=3\,(j+l)+2$, in which case we obtain $n-r=3(p-\,j-\,l)$.

Applying the same reasoning to the sum $B_{n-r}$, there must exist integers $a$ and $b$, such that $m=3\,a+1$ and $n-r-m=3\,b+1$ (satisfying $1\le (3\,a+1)\le n-r-1$). Summing these we obtain $n-r=3\,(a+b)+2$. This cannot hold simultaneously with $n-r$ being a multiple of $3$. So we have a contradiction. A similar contradiction arises when $n=3\,p+3$.

Therefore, when $n$ is either $3\,p+2$ or $3p+3$, for any value of $r$ s.t. $2\le r\le n-2$, a non-zero term occurring in $A_r$ cannot multiply a non-zero term in $B_{n-r}$ . This completes the inductive step.
\end{proof}
\subsection{Proof of Lemma \ref{z large n}}
\begin{proof}
Note that since $b_1$ is positive, if $q$ is real and positive, then all subsequent $b_n$ are also positive. Let $d_p:=b_{3p+1}$. Then the recursion relation \eqref{z formal rec n } can be rewritten as
\begin{align}\label{staggered rec}
d_{p}&=\sum_{i=0}^{p-1}\sum_{j=0}^{i}\sum_{l=1}^{p-1-i}\,d_j\,d_{i-j}\,d_l\,d_{p-1-i-l}\, q^{3\,(i-2j)}
\end{align}
(Here we redefined $m=3\,l+1$, $k=3\,j+1$, $r=3\,i+2$, in order to have non-zero $b_m$, $b_k$, $b_{n-r-m}$  and $b_{r-k}$.) Note that this is symmetric under $q\mapsto 1/q$. 
Define  
\[
c_{p}:=d_{p}\, q^{-\,3\,p\,(p-1)/2}.
\]
Notice that $c_0=1$ and that because $d_1=b_4=1$, we also have $c_1=1$. The recursion relation \eqref{z formal rec n } becomes
\begin{align}\label{scaled rec}
c_{p}&=\sum_{i=0}^{p-1}\sum_{j=0}^{i}\sum_{l=0}^{p-1-i}\,c_j\,c_{i-j}\,c_l\,c_{p-1-i-l}\,\, q^{3\,\bigl(D(i, j, l, p)-p\,(p-1)\bigr)/2}
\end{align}
where $D(i, j, l, p):=j\,(j-1)+(i-j)\,(i-j-1)+2\,(i-2j)+l\,(l-1)+(p-1-i-l)\,(p-1-i-l-1)$. This is a quadratic function of $i$ and $j$. For later reference, we note that
\begin{subequations}\label{D estimates}
\begin{align}
& 1\le j\le p-2\ \Rightarrow\ D(p-1, j, 0, p)-p(p-1)\le -2\,p \\
&1\le l\le p-2\ \Rightarrow\  D(0,0,l,p)-p(p-1)\le -\,4p+6 \\
\nonumber&0\le l\le p-1-i, 0\le j\le i, 1\le i\le p-2 \\
&\qquad\qquad\qquad \Rightarrow\ D(i, j, l, p)-p(p-1)\le -2p+2
\end{align}
\end{subequations}
These are proved by showing that the function is convex with a local minimum in the middle of each respective interval of indices.

We rewrite Equation \eqref{scaled rec} as
\begin{equation}\label{scaled rec again}
\begin{split}
c_{p}&=c_{p-1}\left(1+2\,q^{-3(p-1)}+q^{-6(p-1)}\right)  
\\
&\qquad+\sum_{j=1}^{p-2}\,c_j\,c_{p-1-j}\,q^{3\,\bigl(D(p-1, j, 0, p)-p\,(p-1)\bigr)/2} + \sum_{l=1}^{p-2}\,c_l\,c_{p-1-l}\,q^{3\,\bigl(D(0, 0, l, p)-p\,(p-1)\bigr)/2}\\
&\qquad\qquad +\sum_{i=1}^{p-2}\sum_{j=0}^{i}\sum_{l=0}^{p-1-i}\,c_j\,c_{i-j}\,c_l\,c_{p-1-i-l} q^{3\,\bigl(D(i, j, l, p)-p\,(p-1)\bigr)/2}
\end{split}
\end{equation}
which suggests the transformation
\begin{equation}\label{c to f}
c_p = f_p\,\prod_{k=0}^{p-1}\bigl(1+q^{-3k}\bigr)^2 = f_p\, { (-1;r)_p }^2
\end{equation}
where $r=q^{-3}$ and we have used the notation for the $q$-shifted factorial
\[
(a;q)_n=\left\{\begin{array}{ll}
		(1-a)\,(1-a\,q)\ldots (1-a\,q^{n-1})&n=1, 2, \ldots \\
		1&n=0\\
		\end{array}\right.
\]
Notice that $f_1=1/4$ and $|f_3|<6/5$. Below we use an abbreviation for numbers that are similar to $q$-binomial coefficients
\[
\mathsf{B}^n_k= \left(\frac{ (-1; r)_n}{(-1; r)_k\,(-1; r)_{n-k}}\right)^2
\]
As for usual binomial coefficients, for each fixed $n$, $\mathsf{B}^n_k$ has a local maximum at the half-way point $k=\lfloor n/2\rfloor$. So we get
\begin{eqnarray}\label{q binomial}
\mathsf{B}^n_k &\ge \left\{\begin{array}{l} \mathsf{B}^n_1 =\left( \frac12(1+r^{n-1})\right)^2\ \textrm{for}\ k\ge 1\\
                                                           \mathsf{B}^n_0 =1\ \textrm{for}\ k\ge 0\\
                                                           \end{array}\right.
\end{eqnarray}

The transformation \eqref{c to f} gives 
\begin{equation}\label{scaled rec in f}
\begin{split}
f_{p}&=f_{p-1}\\
&\quad+\frac{1}{(1+q^{-3(p-1)})^2}\Biggl\{ \sum_{j=1}^{p-2} \,\frac{f_j\,f_{p-1-j}}{ \mathsf{B}^{p-1}_j } \,q^{3\,\bigl(D(p-1, j, 0, p)-p\,(p-1)\bigr)/2}\\
&\quad\qquad\qquad +  \sum_{l=1}^{p-2}\,\frac{f_l\,f_{p-1-l}}{ \mathsf{B}^{p-1}_l} \,q^{3\,\bigl(D(0, 0, l, p)-p\,(p-1)\bigr)/2}\\
&\quad\qquad\qquad +\sum_{i=1}^{p-2} \sum_{j=0}^{i}\sum_{l=0}^{p-1-i}\,\frac{f_j\,f_{i-j}}{{ \mathsf{B}^{i}_j} } \frac{f_l\,f_{p-1-i-l}}{{ \mathsf{B}^{p-1-i}_l} } \frac{q^{3\,\bigl(D(i, j, l, p)-p\,(p-1)\bigr)/2}}{{ \mathsf{B}^{p-1}_i} }\Biggr\}
\end{split}
\end{equation}
Assume $|q|\ge e^2$. We now prove inductively that the sequence $\{f_p\}$, for $p\ge 2$, is upper bounded by a constant $\phi$, where $\phi$ is a real root of the quartic $\frac65-\,\phi +\frac32 e^{-4}\phi^2+\frac32 e^{-4}\,\phi^4$, given approximately by $\phi\sim 1.33685$. 

Suppose that $|f_k|\le \phi$ for $k=1, \ldots , p-1$. Using the estimates \eqref{D estimates}, along with  \eqref{q binomial}, we find the upper bound 
\begin{align*}
|f_p|\,-|f_{p-1}|\le |f_p-f_{p-1}|&\le  8\,\phi^2\,\frac{ (p-2)\,\bigl(|q|^{-3(p-2)}+|q|^{-3\,(2p-3)}\bigr)+8\,\phi^2\,p(p-1)(2p-3)|q|^{-3\,(p-1)} }{\bigm|(1+q^{-3(p-1)})^2(1+q^{-3(p-2)})^2\bigm|} \\
                        &\le \frac{3}{2}\,\phi^2\,|q|^{-2(p-2)}\,\frac{ \bigl(1+|q|^{-5\,(p-7/5)}\bigr)+2\,\phi^2\,}{\bigm|(1+q^{-3(p-1)})^2(1+q^{-3(p-2)})^2\bigm|} 
\end{align*}
where we have used the fact that $8(p-2)|q|^{-(p-2)}<8\,e^{-2}<3/2$, $64\,p(p-1)(2p-3)|q|^{-(p-2)}<1152\,e^{-6}<3$ for $p\ge 3$. We sum this inequality over $p$ and use
\begin{align*}
\sum_{k=3}^p\,\frac{|q|^{-2(k-2)}\bigl(1+|q|^{-5\,(k-7/5)}\bigr)}{\bigm|(1+q^{-3(k-1)})^2(1+q^{-3(k-2)})^2\bigm|}
&\le \sum_{k=3}^p\,\frac{|q|^{-2(k-2)}\bigl(1+|q|^{-2\,(k-2)}\bigr)}{(1-|q|^{-2(k-1)})^4}\\
&\le \frac{1}{\log|q|}\,\int_{0}^{q^{-2}}\frac{1+x}{(1-x^2)^4}\,dx\\
&=\frac{1}{6\log|q|}\,\left(\frac{(1+3\,|q|^{-2})}{(1-|q|^{-2})^3}-1\right)\\
&\le \frac{2\,|q|^{-2}}{\log|q|}
\end{align*}
Similarly,
\begin{align*}
\sum_{k=3}^p\,\frac{|q|^{-2(k-2)}}{\bigm|(1+q^{-3(k-1)})^2(1+q^{-3(k-2)})^2\bigm|}
&\le \frac{|q|^{-2}}{2\log|q|}
\end{align*}
Therefore we get
\begin{align*}
|f_p|&\le |f_3|+ \frac32 e^{-4}\phi^2+\frac32 e^{-4}\,\phi^4\le \frac65 +\frac32 e^{-4}\phi^2+\frac32 e^{-4}\,\phi^4 =\phi
\end{align*}
which completes the proof.
\end{proof}
\section{Proofs of Results in \S \ref{true}}\label{app true}
\subsection{Proof of Lemma \ref{perturbation}}
\begin{proof}
Under the transformation $w=W(\xi)+v(\xi)$, Equation \eqref{eq: qp1 scaled} becomes
\begin{align}\label{eq: v eq}
\overline v &= {\mathcal R}(v, \underline v, t)
\end{align}
where $t=1/\xi$ and
\begin{align*}
{\mathcal R}(v, \underline v, t) &=\frac{W+v-t}{(\underline W+\underline v)\,(W+v)^2}-\overline W 
\end{align*}
Note that the function ${\mathcal R}(x, y, t)$ satisfies
\begin{align*}
{\mathcal R}(0, 0, t) &={\mathcal O}\bigl(|t|^N\bigr),\ \forall N\in{\mathbb N}\\
\frac{\partial {\mathcal R}}{\partial x} (0, 0, t)&=\frac{\overline W}{W}\,\left(\frac{1}{\overline W\,W\,\underline W}-2\right)\\
\frac{\partial {\mathcal R}}{\partial y} (0, 0, t)&=-\,\frac{\overline W}{\underline W}
\end{align*}
Rewrite Equation \eqref{eq: v eq} as
\begin{align}\label{eq: v eq 2}
\overline v + \,\left(2\,\frac{\overline W}{W}-\frac{1}{W^2\,\underline W}\right)\,v+\frac{\overline W}{\underline W}\,\underline v &= {\mathcal R}_2(v, \underline v, t)
\end{align}
where 
\begin{align}
{\mathcal R}_2(v, \underline v, t)={\mathcal R}(v, \underline v, t)-\frac{\overline W}{W}\,\left(\frac{1}{\overline W\,W\,\underline W}-2\right)\,v+\frac{\overline W}{\underline W}\,\underline v
\end{align}
has a Taylor expansion that starts with an arbitrarily small constant term and quadratic terms in $v$, $\underline v$, for $(v, \underline v, t)\in \mathcal D$, for some non-zero $\epsilon_0$ and $\delta_0$. The estimates for ${\mathcal R}_2$ follow from this fact.
\end{proof}
\subsection{Proof of Lemma \ref{linear eq}}
\begin{proof}
Because we have $W\sim 1/\xi$, as $|\xi|\to\infty$, the largest coefficient in Equation \ref{eq: lin hom P} is $1/(W^2\,\underline W)$. We use this term as the starting point of two recursive processes as follows
\begin{align*}
&\overline P\,\sim \frac{P}{W^2\,\underline W}+{\mathcal O}(P,\,\underline P)\\
\Rightarrow\ 
P^+(\xi)&\sim \prod_{j=n_0}^n\frac{1}{W_j^2\,W_{j-1}}
\end{align*}
and
\begin{align*}
&P\,\sim \underline P\,W^2\,\overline W+{\mathcal O}(P,\,\overline P)\\
\Rightarrow\ 
P^-(\xi)&\sim \prod_{j=n_0}^n\,W_{j+1}\,W_j^2
\end{align*}
We get the remaining results in Equations \eqref{pasym} by applying the asymptotic methods provided in Chapter 5 \cite{bo:78} for irregular-singular-type limits of linear difference equations. In particular, the transformation of variables: $P^{\pm}=q^{\pm 3 n^2/2}\,R^{\pm}$ leads to the leading order results: $R^+\sim \rho_+\,q^{-5\,n/2}$, $R^-\sim \rho_-\,q^{-5\,n/2}$ for some constants $\rho_{\pm}$. Applying standard theorems \cite{hs:64}, we get true solutions with these behaviours, which we continue to denote by $P^{\pm}$.
\end{proof}

\subsection{Proof of Lemma \ref{summation eq}}
\begin{proof}

Multiplying Equation \eqref{eq: v eq 2} by $P$, a solution of Equation \eqref{eq: lin hom P}, we get
\[
\frac{P\,\overline v}{\overline W\,W}-\,\frac{\underline P\, v}{W\,\underline W}
-\,\frac{\overline P\, v}{\overline W\,W}+\,\frac{ P\, \underline v}{W\,\underline W}=\frac{P\,{\mathcal R}_2}{\overline W\,W}
\]
Since the left side is an exact difference, we integrate (i.e., sum) to get
\begin{equation}\label{forward sum 1}
\frac{\underline P\, v}{W\,\underline W}-\,\frac{P\,\underline v}{W\,\underline W}=\sum_{k=k_0}^{n-1}\frac{P_k\,{\mathcal R}_2(v_k, v_{k-1},t_k)}{W_{k+1}\,W_{k}}+\left(\frac{\underline P\, v}{W\,\underline W}-\,\frac{P\,\underline v}{W\,\underline W}\right)\Bigm|_{n=k_0}
\end{equation}
Notice that we could equally well have summed in the other direction, to get
\begin{equation}\label{forward sum 2}
\frac{\underline P\, v}{W\,\underline W}-\,\frac{P\,\underline v}{W\,\underline W}=-\,\sum_{k=n}^{m_0-1}\frac{P_k\,{\mathcal R}_2(v_k, v_{k-1},t_k)}{W_{k+1}\,W_{k}}+\left(\frac{\underline P\, v}{W\,\underline W}-\,\frac{P\,\underline v}{W\,\underline W}\right)\Bigm|_{n=m_0}
\end{equation}

Multiplying by $W\,\underline W$ and dividing by $P\,\underline P$, we again get an exact difference on the left, which we integrate to obtain Equation \eqref{eq: v int}.
\end{proof}
\subsection{Proof of Lemma \ref{contraction}}
\begin{proof}
Let $\mathcal D$ be the set of all $({\mathbf v}, t) \in \complex^2\times\complex$ where ${\mathbf v}=(x, y)$, such that $\|{\mathbf v}\|\le \delta_0$, $|t|<\epsilon_0$. Also let ${\mathcal V}$ denote the set of all continuous functions ${\mathbf v}: {\mathcal T}\to \complex^2$ such that $(x, y, t)\in {\mathcal D}$ for every $t\in{\mathcal T}$ and $t\mapsto {\mathcal R}_2(x(t), y(t), t)$ is bounded on ${\mathcal T}$. 

We let $n_0$ run to infinity in the summation equation \eqref{eq: v int reverse} and get
\begin{align}\label{eq: v to infty}
v^+_n=-\,P^+_n\,\sum_{j=n}^{\infty}\,\frac{W_j\,W_{j-1}}{P^+_j\,P^+_{j-1}}\,\sum_{k=k_0}^{j-1}\frac{P^+_k\,{\mathcal R}_2(v_k, v_{k-1},t_k)}{W_{k+1}\,W_{k}}
\end{align}
where $P^+$ is as identified in Lemma \ref{linear eq}. Recall that in $\mathcal S$, we have $|P^+|\gg |P^-|$. Here, note that the functions ${P^+_k}/(W_{k+1}\,W_{k})$ in the inner sum (in $k$) are largest at the upper limit $k=j-1$ of this sum. But the functions  ${W_j\,W_{j-1}}/(P^+_j\,P^+_{j-1})$ in the outer sum are much much smaller, so that the product of the inner and outer sums is asymptotic to $|q|^{-3\,j^2/2+5j/2}$.

Let the right side denote the action of an operator ${\mathcal F}$ acting on each ${\mathbf v}\in{\mathcal V}$. Let $V$ be the space of all continuous functions ${\mathbf v}:{\mathcal T}\to \complex^2$ such that $\|{\mathbf v}\|\le \delta$ for all $t\in {\mathcal T}$. Then $V$ is a complete metric space with $d(v, w):=sup_{t\in {\mathcal T}}\|v(t)-w(t)\|$ and we get
\begin{align}\label{estimates}
|{\mathcal F}(v)(\xi)|&\le \,\sum_{m=2}^{\infty}\,|q|^{-3m(m-5/3)/2}\,(C_1\,\delta^2+C_2\,\epsilon)\le \delta\\
\left|{\mathcal F}(v^{(1)})(\xi)-{\mathcal F}(v^{(2)})(\xi)\right|&\le \,\sum_{m=2}^{\infty}\,|q|^{-3m(m-5/3)/2}\,(C_3\,\delta+C_4\,\epsilon)\|v^{(1)}_{n_0+m}-v^{(2)}_{n_0+m}\|
\end{align}
for every $v$, $v^{(1)}$, $v^{(2)}$ $\in { V}$ and $t\in {\mathcal T}$. It follows that ${\mathcal F}({ V})\subset { V}$ and that ${\mathcal F}$ is a contraction with contraction factor that is bounded above by $|q|\bigl(C_3\,\delta+C_4\,\epsilon\bigr)<1$, which completes the proof of the lemma.
\end{proof}

\end{document}